\documentclass[submission,copyright,creativecommons]{eptcs}
\providecommand{\event}{ICLP 2024} 

\usepackage{iftex}

\ifpdf
  \usepackage{underscore}         
  \usepackage[T1]{fontenc}        
\else
  \usepackage{breakurl}           
\fi

\usepackage{subcaption}

\usepackage{amsthm} 
\usepackage{amsmath}
\usepackage{graphicx}
\usepackage{multirow}

\usepackage{amssymb}

\usepackage{tikz}
\usetikzlibrary{arrows,shapes,automata,petri,positioning}

\tikzset{
	place/.style={
		circle,
		thick,
		draw=blue!75,
		fill=blue!20,
		minimum size=6mm,
	},
	transitionH/.style={
		rectangle,
		thick,
		fill=black,
		minimum width=8mm,
		inner ysep=2pt
	},
	transitionV/.style={
		rectangle,
		thick,
		fill=black,
		minimum height=8mm,
		inner xsep=2pt
	}
}

\usepackage{forest}
\usepackage{tikz-qtree}

\newcommand{\dng}[1]{{{\sim}#1}}
\newcommand{\dg}[1]{\text{dg}({#1})}
\newcommand{\pdg}[1]{\text{dg}^+({#1})}
\newcommand{\ig}[1]{\text{ig}({#1})}

\newcommand{\atom}[1]{\text{atom}({#1})}
\newcommand{\head}[1]{\text{h}({#1})}

\newcommand{\bodyf}[1]{\text{bf}({#1})}
\newcommand{\pbody}[1]{\text{b}^+({#1})}
\newcommand{\nbody}[1]{\text{b}^-({#1})}

\newcommand{\cf}[1]{\text{cf}({#1})}
\newcommand{\cfl}[1]{\overleftarrow{\text{cf}}({#1})}

\newcommand{\var}[1]{\text{var}_{{#1}}}

\newcommand{\lfp}[1]{\text{lfp}({#1})}
\newcommand{\tgsp}[1]{\text{tg}_{sp}({#1})}
\newcommand{\tgst}[1]{\text{tg}_{st}({#1})}
\newcommand{\rhs}[1]{\text{rhs}({#1})}

\newcommand{\stg}[1]{\text{sstg}({#1})}
\newcommand{\atg}[1]{\text{astg}({#1})}

\newtheorem{theorem}{Theorem}
\newtheorem{lemma}[theorem]{Lemma}

\newtheorem{proposition}[theorem]{Proposition}

\newtheorem{definition}{Definition}
\newtheorem{example}{Example}

\providecommand{\keywords}[1]
{
	\small	
	\textbf{\textit{Keywords:}} #1
}

\title{Graphical Conditions for the Existence, Unicity and Number of Regular Models}
\author{Van-Giang Trinh \qquad\qquad Belaid Benhamou
\institute{LIRICA team, LIS, Aix-Marseille University, Marseille, France}
\email{\quad van-giang.trinh@inria.fr \quad\qquad belaid.benhamou@lis-lab.fr}
\and
Sylvain Soliman \qquad\qquad Fran\c{c}ois Fages
\institute{Inria Saclay, EP Lifeware, Palaiseau, France}
\email{\quad Sylvain.Soliman@inria.fr \quad\qquad francois.fages@inria.fr}
}
\def\titlerunning{Graphical Conditions for the Existence, Unicity and Number of Regular Models}
\def\authorrunning{V.G. Trinh, B. Benhamou, S. Soliman \& F. Fages}
\begin{document}
\maketitle

\begin{abstract}
The regular models of a normal logic program are a particular type of partial (i.e.~3-valued) models which correspond to stable partial models with minimal undefinedness.
In this paper, we explore graphical conditions on the dependency graph of a finite ground normal logic program to analyze the existence, unicity and number of regular models for the program.
We show three main results: 1) a necessary condition for the existence of non-trivial (i.e.~non-2-valued) regular models, 2) a sufficient condition for the unicity of regular models, and 3) two upper bounds for the number of regular models based on positive feedback vertex sets.
The first two conditions generalize the finite cases of the two existing results obtained by You and Yuan (1994)
for normal logic programs with well-founded stratification.
The third result is also new to the best of our knowledge.
Key to our proofs is a connection that we establish between finite ground normal logic programs and Boolean network theory.

\keywords{logic programming, semantics of negation, canonical model, three-valued model, Datalog, abstract argumentation, Boolean network, feedback vertex set, model counting}
\end{abstract}

\section{Introduction}

Relating graphical representations of a \emph{normal logic program} (or just \emph{program} if not otherwise said) and its model-theoretic semantics is an interesting research direction in theory that also has many useful applications in practice~\cite{cois1994consistency,costantini2006existence,DBLP:conf/ijcai/Linke01}.
Historically, the first studies of this direction focused on the existence of a unique stable model in classes of programs with special graphical properties on (positive) dependency graphs, including positive programs~\cite{gelfond1988stable}, acyclic programs~\cite{apt1991acyclic}, and locally stratified programs~\cite{gelfond1988stable}.
In 1991, Fages gave a simple characterization of stable models as well-supported models in~\cite{DBLP:journals/ngc/Fages91},
and then showed that for \emph{tight} programs (i.e.~without non-well-founded positive justifications),
the stable models of the program coincide with the Herbrand models of its Clark's completion~\cite{cois1994consistency}.
Being finer-represented but more computationally expensive than dependency graphs, several other graphical representations (e.g., cycle and extended dependency graphs, rule graphs, block graphs) were introduced and several improved results were obtained~\cite{costantini2006existence,DBLP:conf/gkr/CostantiniP11,DBLP:journals/tcs/DimopoulosT96,DBLP:conf/ijcai/Linke01}.
There are some recent studies on dependency graphs~\cite{DBLP:journals/tplp/FandinnoL23,TB24-static-analysis}, but they still focus only on stable models.
In contrast, very few studies were made about regular models despite of their prominent importance in argumentation frameworks~\cite{DBLP:journals/sLogica/WuCG09,DBLP:journals/jair/CaminadaS17}
and program semantics~\cite{DBLP:journals/tocl/JanhunenNSSY06}.
The work of~\cite{DBLP:journals/amai/EiterLS97} showed the unicity of regular and stable models in locally stratified programs.
The work of~\cite{DBLP:journals/jcss/YouY94} showed two sufficient graphical conditions, one for the coincidence between stable and regular models, and another one for the unicity of regular models.
However, these two conditions were only proven in the case of well-founded stratification programs, and the question if they are still valid for any program is still open to date.

The stable partial semantics is the 3-valued generalization of the (2-valued) stable model semantics~\cite{DBLP:journals/fuin/Przymusinski90}.
The regular model semantics not only inherits the advantages of the stable partial model semantics but also imposes two notable principles in non-monotonic reasoning: \emph{minimal undefinedness} and \emph{justifiability} (which is closely related to the concept of labeling-based justification in Doyle's truth maintenance system~\cite{DBLP:journals/ai/Doyle79}), making it become one of the well-known semantics in logic programming~\cite{DBLP:journals/jcss/YouY94,DBLP:journals/tocl/JanhunenNSSY06}.
Furthermore, regular models in ground programs were proven to correspond to preferred extensions in Dung's frameworks~\cite{DBLP:journals/sLogica/WuCG09} and assumption-based argumentation~\cite{DBLP:journals/jair/CaminadaS17}, which are two central focuses in abstract argumentation~\cite{DBLP:journals/ker/BaroniCG11}.

Recently, we have proposed a new semantics for finite ground programs, called the \emph{trap space semantics}, which establishes formal links between the model-theoretic and dynamical semantics of a finite ground program~\cite{TBFS24lpnmr-semantics}.
It is built on two newly proposed concepts: \emph{stable} and \emph{supported trap spaces}, which are inspired by the concepts of \emph{trap} (or its duality, \emph{siphon}) in Petri net theory and \emph{trap space} in Boolean network theory~\cite{Murata1989,DBLP:journals/nc/KlarnerBS15,TRINH2023114073,DBLP:journals/corr/abs-2403-06255}.
We relate the new semantics to other widely-known semantics, in particular showing that subset-minimal stable trap spaces of a finite ground program coincide with its regular models.
Interestingly, the restriction to finite ground programs applies without loss of generality
to normal Datalog programs, i.e.~normal logic programs built over an alphabet without function symbols,
since their Herbrand base and their ground instanciation are finite~\cite{ceri1990logic}.

Motivated by the above elements, in this paper, we explore graphical conditions on the dependency graph of a finite ground program to analyze the existence of non-trivial (i.e.~not 2-valued) regular models and the unicity and multiplicity of regular models for the program.
More specifically, we show three main results: 1) the existence of negative cycles is a necessary condition for the existence of non-trivial regular models, 2) the absence of positive cycles is a sufficient condition for the unicity of regular models, and 3) \(3^{|U^{+}|}\) (resp.\ \(2^{|U^{+}|}\)) is an upper bound (resp.\ a finer upper bound) for the number of regular models in generic (resp.\ tight) finite ground programs where \(U^{+}\) is a positive feedback vertex set of the dependency graph.
The first two conditions generalize the finite cases of the two existing results obtained by~\cite{DBLP:journals/jcss/YouY94} for well-founded stratification normal logic programs.
The third result is also new to the best of our knowledge.
Key to our proofs is a connection that we establish between finite ground programs and Boolean network theory based on the trap space semantics.

Boolean Networks (BNs) are a simple and efficient mathematical formalism that has been widely applied to many areas from science to engineering~\cite{schwab2020concepts}.
Originated in the early work of~\cite{thomas1990biological}, studying relationships between the dynamics of a BN and its influence graph has a rich history of research~\cite{DBLP:journals/entcs/PauleveR12,richard2019positive}.
To date, this research direction is still growing with many prominent and deep results~\cite{richard2019positive,schwab2020concepts,richard2023attractor}.
Hence, the established connection can bring a plenty of existing results in BNs to studying finite ground programs as well as provide a unified framework for exploring and proving more new theoretical results in the logic program theory.

The rest of this paper is organized as follows.
In the next section, we recall preliminaries on normal logic programs, regular models, BNs, and related concepts.
Section~\ref{sec:connection} presents the connection that we establish between finite ground programs and BNs.
In Section~\ref{sec:main-results}, we present the main results on relationships between regular models and graphical conditions.
Finally, Section~\ref{sec:conclusion} concludes the paper with some perspectives for future work.

\section{Preliminaries}\label{sec:preliminaries}

We assume that the reader is familiar with the logic program theory and the stable model semantics~\cite{gelfond1988stable}.
Unless specifically stated, a \emph{program} means a normal logic program.
In addition, we consider the Boolean domain \(\mathbb{B} = \{\text{true}, \text{false}\} = \{1, 0\}\),
and the Boolean connectives used in this paper include \(\land\) (conjunction), \(\lor\) (disjunction), \(\neg\) (negation), \(\leftarrow\)
(implication), and \(\leftrightarrow\) (equivalence).

\subsection{Normal logic programs}

We consider a first-order language built over an infinite alphabet of variables,
and finite alphabets of constant, function and predicates symbol. 
The set of first-order \emph{terms} is the least set containing variables, constants and closed by application of function symbols.
An \emph{atom} is a formula of the form \(p(t_1, \dots, t_k)\) where \(p\) is a predicate symbol and \(t_i\) are terms.
A \emph{normal logic program} \(P\) is a \emph{finite} set of \emph{rules} of the form
\[p \gets p_1, \dots, p_m, \dng{p_{m + 1}}, \dots, \dng{p_{k}}\] where \(p\) and \(p_i\) are atoms, \(k \geq m \geq 0\),
and \(\sim\) is a symbol for negation. 
A fact is a rule with \(k = 0\). 
We denote by \(\atom{P}\) the set of atoms appearing in \(P\).
For any rule \(r\) of the above form, \(\head{r} = p\) is the \emph{head} of \(r\), \(\pbody{r} = \{p_1, \dots, p_m\}\) is called the \emph{positive body} of \(r\), \(\nbody{r} = \{p_{m + 1}, \dots, p_{k}\}\) is called the \emph{negative body} of \(r\), and \(\bodyf{r} = p_1 \land \dots \land p_m \land \neg p_{m + 1} \land \dots \land \neg p_{k}\) is the \emph{body formula} of \(r\).
If \(\nbody{r} = \emptyset, \forall r \in P\), then \(P\) is called a \emph{positive program}.
If \(\pbody{r} = \emptyset, \forall r \in P\), then \(P\) is called a \emph{quasi-interpretation} program.

A term, an atom or a program is \emph{ground} if it contains no variable.
The \emph{Herbrand base} is the set of ground atoms formed over the alphabet of the program.
It is finite in absence of function symbol, which is the case of \emph{Datalog} programs~\cite{ceri1990logic}.
The \emph{ground instantiation} of a program \(P\) is the set of the ground instances of all rules in \(P\).
In the rest of the paper, we restrict ourselves to \emph{finite ground normal logic programs}.

We shall use the fixpoint semantics of normal logic programs~\cite{dung1989fixpoint} to prove many new results in the next sections.
To be self-contained, we briefly recall the definition of the \emph{least fixpoint} of a normal logic program \(P\) as follows.
Let \(r\) be the rule \(p \leftarrow \dng{p_1}, \dots, \dng{p_k}, q_1, \dots, q_j\) and let \(r_i\) be rules \(q_i \leftarrow \dng{q^1_i}, \dots, \dng{q^{l_i}_i}\) where \(1 \leq i \leq j\) and \(l_i \geq 0\).
Then \(\sigma_{r}(\{r_1, \dots, r_j\})\) is the following rule \[p \leftarrow \dng{p_1}, \dots, \dng{p_k}, \dng{q_1^1}, \dots, \dng{q_1^{l_1}}, \dots, \dng{q_j^1}, \dots, \dng{q_j^{l_j}}.\]
\(\sigma_P\) is the transformation on quasi-interpretation programs: \(\sigma_P(Q) = \{\sigma_r(\{r_1, \dots, r_j\}) | r \in P, r_i \in Q, 1 \leq i \leq j\}\).
Let \(\text{lfp}_i = \sigma_P^i(\emptyset) = \sigma_P(\sigma_P(\dots \sigma_P(\emptyset)))\), then \(\lfp{P} = \bigcup_{i \geq 1}\text{lfp}_i\) is the least fixpoint of \(P\).
In the case of finite ground programs, \(\lfp{P}\) is finite and also a quasi-interpretation finite ground program~\cite{dung1989fixpoint}.

\subsubsection{Stable and supported partial models}

A \emph{3-valued interpretation} \(I\) of a finite ground program \(P\) is a total function \(I \colon \atom{P} \rightarrow \{\textbf{t}, \textbf{f}, \textbf{u}\}\) that assigns one of the truth values true (\textbf{t}), false (\textbf{f}) or unknown (\textbf{u}), to each atom of \(P\).
If \(I(a) \neq \textbf{u}, \forall a \in \atom{P}\), then \(I\) is an \emph{Herbrand (2-valued) interpretation} of \(P\).
Usually, a 2-valued interpretation is written as the set of atoms that are true in this interpretation.
A 3-valued interpretation \(I\) characterizes  the set of 2-valued interpretations denoted by \(\gamma(I)\) as \(\gamma(I) = \{J | J \in 2^{\text{atom}(P)}, \forall a \in \atom{P}, I(a) \neq \textbf{u} \Rightarrow J(a) = I(a)\}\).
For example, if \(I = \{p = \textbf{t}, q = \textbf{f}, r = \textbf{u}\}\), then \(\gamma(I) = \{\{p\}, \{p, r\}\}\).

We consider two orders on 3-valued interpretations.
The truth order \(\leq_{t}\) is given by \(\textbf{f} <_{t} \textbf{u} <_{t} \textbf{t}\).
Then, \(I_1 \leq_t I_2\) iff \(I_1(a) \leq_{t} I_2(a), \forall a \in \atom{P}\).
The subset order \(\leq_{s}\) is given by \(\textbf{f} <_{s} \textbf{u}\) and \(\textbf{t} <_{s} \textbf{u}\).
Then, \(I_1 \leq_{s} I_2\) iff \(I_1(a) \leq_{s} I_2(a), \forall a \in \atom{P}\). In addition, \(I_1 \leq_{s} I_2\) iff \(\gamma(I_1) \subseteq \gamma(I_2)\), i.e., \(\leq_{s}\) is identical to the subset partial order.

Let \(f\) be a propositional formula on \(\atom{P}\).
Then the valuation of \(f\) under a 3-valued interpretation \(I\) (denoted by \(I(f)\)) is defined recursively as follows:
\begin{align*}
I(f) = \begin{cases}
I(a) &\text{if } f = a, a \in \atom{P}\\
\neg I(f_1) &\text{if } f = \neg f_1\\
\text{min}_{\leq_t}(I(f_1), I(f_2)) &\text{if } f = f_1 \land f_2\\
\text{max}_{\leq_t}(I(f_1), I(f_2)) &\text{if } f = f_1 \lor f_2
\end{cases}
\end{align*} where \(\neg \textbf{t} = \textbf{f}, \neg \textbf{f} = \textbf{t}, \neg \textbf{u} = \textbf{u}\), and \(\text{min}_{\leq_t}\) (resp.\ \(\text{max}_{\leq_t}\)) is the function to get the minimum (resp.\ maximum) value of two values w.r.t.\ the order \(\leq_t\).
We say 3-valued interpretation \(I\) is a \emph{3-valued model} of a finite ground program \(P\) iff for each rule \(r \in P\), \(I(\bodyf{r}) \leq_{t} I(\head{r})\).

\begin{definition}
	Let \(I\) be a 3-valued interpretation of \(P\).
	We build the \emph{reduct} \(P^I\) as follows.
	\begin{itemize}
		\item Remove any rule \(a \leftarrow a_1, \dots, a_m, \dng{b_1}, \dots, \dng{b_k} \in P\) if \(I(b_i) = \textbf{t}\) for some \(1 \leq i \leq k\).
		\item Afterwards, remove any occurrence of \(\dng{b_i}\) from \(P\) such that \(I(b_i) = \textbf{f}\).
		\item Then, replace any occurrence of \(\dng{b_i}\) left by a special atom \textbf{u} (\(\textbf{u} \not \in \atom{P}\)).
	\end{itemize}
	\(P^I\) is positive and has a unique \(\leq_{t}\)-least 3-valued model.
	See~\cite{DBLP:journals/fuin/Przymusinski90} for the method for computing this \(\leq_{t}\)-least 3-valued model.
	Then \(I\) is a \emph{stable partial model} of \(P\) iff \(I\) is equal to the \(\leq_{t}\)-least 3-valued model of \(P^I\).
	A stable partial model \(I\) is a regular model if it is \(\leq_s\)-minimal.
	A regular model is non-trivial if it is not 2-valued.
\end{definition}

The \emph{Clark's completion} of a finite ground program \(P\) (denoted by \(\cf{P}\)) consists of the following sentences: for each \(p \in \atom{P}\), let \(r_1, \dots , r_k\) be all the rules of \(P\) having the same head \(p\), then \(p \leftrightarrow \bodyf{r_1} \lor \dots \lor \bodyf{r_k}\) is in \(\cf{P}\).
If there is no rule whose head is \(p\), then the equivalence is \(p \leftrightarrow \textbf{f}\).
Let \(\rhs{a}\) denote the right hand side of atom \(a\) in \(\cf{P}\).
A 3-valued interpretation \(I\) is a 3-valued model of \(\cf{P}\) iff for every \(a \in \atom{P}\), \(I(a) = I(\rhs{a})\).
We define a \emph{supported partial model} of \(P\) as a 3-valued model of \(\cf{P}\).
Note that 2-valued stable (resp.\ supported) partial models are stable (resp.\ supported) models.

\subsubsection{Dependency and transition graphs}

The Dependency Graph (DG) of a finite ground program \(P\) (denoted by \(\dg{P}\)) is a signed directed graph \((V, E)\) on the set of signs \(\{\oplus, \ominus\}\) where \(V = \atom{P}\) and \((uv, \oplus) \in E\) (resp.\ \((uv, \ominus) \in E\)) iff there is a rule \(r \in P\) such that \(v = \head{r}\) and \(u \in \pbody{r}\) (resp.\ \(u \in \nbody{r}\)).
An arc \((uv, \oplus)\) is positive, whereas an arc \((uv, \ominus)\) is negative.
Since \(\atom{P}\) is finite, the DG of \(P\) is a finite graph, thus we can apply the finite graph theory.
A cycle of \(\dg{P}\) is positive (resp.\ negative) if it contains an even (resp.\ odd) number of negative arcs.
A positive (resp.\ negative) feedback vertex set is a set of vertices that intersect all positive (resp.\ negative) cycles of \(\dg{P}\).
The positive DG of \(P\) (denoted by \(\pdg{P}\)) is a sub-graph of \(\dg{P}\) that has the same set of vertices but contains only positive arcs.
\(P\) is \emph{locally stratified} if every cycle of \(\dg{P}\) contains no negative arc~\cite{gelfond1988stable}.
\(P\) is \emph{tight} if \(\pdg{P}\) has no cycle~\cite{cois1994consistency}.
\(P\) is \emph{well-founded stratification} if there is a topological order on the set of Strongly Connected Components (SCCs) of \(\dg{P}\) and for every SCC \(B\), there exists SCC \(A \leq B\) and for any SCC \(C\), if \(C \leq A\) then there are only positive arcs from atoms in \(C\) to atoms in \(A\)~\cite{DBLP:journals/jcss/YouY94}.
Herein, \(A \leq B\) iff there is a path from some atom in \(A\) to some atom in \(B\).
In the case of finite ground programs, the above definition of well-founded stratification (which was orginally defined for both finite and infinite ground programs) is equivalent to that a finite ground program is well-founded stratification iff there is a topological order of its dependency graph such that every SCC at the lowest level only contains positive arcs.

The \emph{immediate consequence operator} (or the \emph{\(T_P\) operator}) is defined as a mapping \(T_P \colon 2^{\atom{P}} \to 2^{\atom{P}}\) such that \(T_P(I)(a) = I(\rhs{a})\) where \(I\) is a 2-valued interpretation.
If \(I\) is a 2-valued interpretation, then \(P^I\) is exactly the reduct defined in~\cite{gelfond1988stable} and the unique \(\leq_t\)-least model of \(P^I\) is 2-valued.
The \emph{Gelfond-Lifschitz operator} (or the \emph{\(F_P\) operator}) is defined as a mapping \(F_P \colon 2^{\atom{P}} \to 2^{\atom{P}}\) such that \(F_P(I)\) is the unique \(\leq_t\)-least model of \(P^I\)~\cite{gelfond1988stable}.
The \emph{stable} (resp.\ \emph{supported}) \emph{transition graph} of \(P\) is a directed graph (denoted by \(\text{tg}_{st}(P)\) (resp.\ \(\text{tg}_{sp}(P)\))) on the set of all possible 2-valued interpretations of \(P\) such that \((I, J)\) is an arc of \(\text{tg}_{st}(P)\) (resp.\ \(\text{tg}_{sp}(P)\)) iff \(J = F_P(I)\) (resp.\ \(J = T_P(I)\)).
A \emph{trap domain} of a directed graph is a set of vertices having no out-going arcs.

\subsubsection{Stable and supported trap spaces}

In~\cite{TBFS24lpnmr-semantics}, we introduce a new semantics for finite ground programs, called the \emph{trap space semantics}.
This semantics shall be used in this work as the bridge between finite ground programs and Boolean networks.
To be self-contained, we briefly recall the definition and essential properties of this semantics.

A set \(S\) of 2-valued interpretations of a finite ground program \(P\) is called a \emph{stable trap set} (resp.\ \emph{supported trap set}) of \(P\) if \(\{F_P(I) | I \in S\} \subseteq S\) (resp.\ \(\{T_P(I) | I \in S\} \subseteq S\)).
A 3-valued interpretation \(I\) of a finite ground program \(P\) is called a \emph{stable trap space} (resp.\ \emph{supported trap space}) of \(P\) if \(\gamma(I)\) is a stable (resp.\ supported) trap set of \(P\).
By definition, a stable (resp.\ supported) trap set of \(P\) is a trap domain of \(\tgst{P}\) (resp.\ \(\tgsp{P}\)).
Hence, we can deduce that a 3-valued interpretation \(I\) is a stable (resp.\ supported) trap space of \(P\) if \(\gamma(I)\) is a trap domain of \(\tgst{P}\) (resp.\ \(\tgsp{P}\)).
We also show in~\cite{TBFS24lpnmr-semantics} that \(I\) is a supported trap space of \(P\) iff \(I\) is 3-valued model of \(\cfl{P}\) w.r.t.\ to the order \(\leq_s\) where \(\cfl{P}\) is the \(\leftarrow\) part of the Clark's completion of \(P\), and a stable (resp.\ supported) partial model of \(P\) is also a stable (resp.\ supported) trap space of \(P\).

\begin{example}\label{example:lp}
	Consider finite ground program \(P_1\) (taken from~\cite{DBLP:conf/birthday/InoueS12}) where \(P_1 = \{p \leftarrow \dng{q}; q \leftarrow \dng{p}; r \leftarrow q\}\).
	Herein, we use ';' to separate program rules. 
	Figures~\ref{figure:dependency-stable-supported-transition-graphs}~(a),~(b), and~(c) show the dependency graph, the stable transition graph, and the supported transition graph of \(P_1\), respectively.
	\(P_1\) is tight, but neither locally stratified nor well-founded stratification.
	\(P_1\) has five stable (also supported) trap spaces: \(I_1 = \{p = \textbf{t}, q = \textbf{f}, r = \textbf{u}\}\), \(I_2 = \{p = \textbf{f}, q = \textbf{t}, r = \textbf{u}\}\), \(I_3 = \{p = \textbf{u}, q = \textbf{u}, r = \textbf{u}\}\), \(I_4 = \{p = \textbf{t}, q = \textbf{f}, r = \textbf{f}\}\),  and \(I_5 = \{p = \textbf{f}, q = \textbf{t}, r = \textbf{t}\}\).
	Among them, only \(I_3\), \(I_4\), and \(I_5\) are stable (also supported) partial models of \(P_1\).
	\(P_1\) has two regular models (\(I_4\) and \(I_5\)).
	The least fixpoint of \(P_1\) is \(\lfp{P_1} = \{p \leftarrow \dng{q}; q \leftarrow \dng{p}; r \leftarrow \dng{p}\}\).
\end{example}

\begin{figure}[!ht]
	\centering
	\begin{subfigure}{0.25\textwidth}
		\centering
		\begin{tikzpicture}[node distance=1cm and 1cm, every node/.style={scale=0.8}, line width = 0.5mm]
		\node[circle, draw] (p) [] {$p$};
		\node[circle, draw] (q) [below=of p] {$q$};
		\node[circle, draw] (r) [below=of q] {$r$};
		
		\draw[->] (q) edge [bend right=30] node [midway, above, fill=white] {$\ominus$} (p);
		\draw[->] (p) edge [bend right=30] node [midway, above, fill=white] {$\ominus$} (q);
		\draw[->] (q) edge [] node [midway, above, fill=white] {$\oplus$} (r);
		\end{tikzpicture}
		\caption{}
	\end{subfigure}
	\begin{subfigure}{0.35\textwidth}
		\centering
		\begin{tikzpicture}[node distance=1cm and 1cm, every node/.style={scale=0.8}, line width = 0.5mm]
		\node[] (pr) [] {$\{p, r\}$};
		\node[] (p) [below=of pr] {$\{p\}$};
		\node[] (q) [right=of pr] {$\{q\}$};
		\node[] (qr) [below=of q] {$\{q, r\}$};
		\node[] (e) [below=of p] {$\emptyset$};
		\node[] (pqr) [below=of qr] {$\{p, q, r\}$};
		\node[] (r) [right=of pqr] {$\{r\}$};
		\node[] (pq) [right=of qr] {$\{p, q\}$};
		
		\draw[->] (pr) edge [] (p);
		\draw[->] (q) edge [] (qr);
		
		\draw[->] (pq) edge [] (e);
		\draw[->] (r) edge [] (pqr);
		
		\draw[->] (e) edge [] (pqr);
		\draw[->] (pqr) edge [bend left=20] (e);
		
		\path [] (p) edge [out=0, in=45, loop] (p);
		\path [] (qr) edge [out=0, in=45, loop] (qr);
		\end{tikzpicture}
		\caption{}
	\end{subfigure}
	\begin{subfigure}{0.35\textwidth}  
		\centering
		\begin{tikzpicture}[node distance=1cm and 1cm, every node/.style={scale=0.8}, line width = 0.5mm]
		\node[] (pr) [] {$\{p, r\}$};
		\node[] (p) [below=of pr] {$\{p\}$};
		\node[] (q) [right=of pr] {$\{q\}$};
		\node[] (qr) [below=of q] {$\{q, r\}$};
		\node[] (e) [below=of p] {$\emptyset$};
		\node[] (pqr) [below=of qr] {$\{p, q, r\}$};
		\node[] (r) [right=of pqr] {$\{r\}$};
		\node[] (pq) [right=of qr] {$\{p, q\}$};
		
		\draw[->] (pr) edge [] (p);
		\draw[->] (q) edge [] (qr);
		\draw[->] (e) edge [] (pq);
		\draw[->] (pqr) edge [] (r);
		\draw[->] (pq) edge [bend left=15] (r);
		\draw[->] (r) edge [bend left=15] (pq);
		
		\path [] (p) edge [out=0, in=45, loop] (p);
		\path [] (qr) edge [out=0, in=45, loop] (qr);
		\end{tikzpicture}
		\caption{}
	\end{subfigure}
	\caption{(a) \(\dg{P_1}\), (b) \(\tgst{P_1}\), and (c) \(\tgsp{P_1}\).}\label{figure:dependency-stable-supported-transition-graphs}
\end{figure}

\subsection{Boolean networks}

A Boolean Network (BN) \(f\) is a \emph{finite} set of Boolean functions on a finite set of Boolean variables denoted by \(\var{f}\).
Each variable \(v\) is associated with a Boolean function \(f_v \colon \mathbb{B}^{|\var{f}|} \rightarrow \mathbb{B}\).
\(f_v\) is called \emph{constant} if it is always either 0 or 1 regardless of its arguments.
A state \(s\) of \(f\) is a mapping \(s \colon \var{f} \mapsto \mathbb{B}\) that assigns either 0 (inactive) or 1 (active) to each variable.
We can write \(s_v\) instead of \(s(v)\) for short.

Let \(x\) be a state of \(f\).
We use \(x[v \leftarrow a]\) to denote the state \(y\) so that \(y_v = a\) and \(y_u = x_u, \forall u \in \var{f}, u \neq v\) where \(a \in \mathbb{B}\).
The Influence Graph (IG) of \(f\) (denoted by \(\ig{f} \)) is a signed directed graph \((V, E)\) on the set of signs \(\{\oplus, \ominus\}\) where \(V = \var{f}\), \((uv, \oplus) \in E\) (i.e., \(u\) positively affects the value of \(f_v\)) iff there is a state \(x\) such that \(f_v(x[u \leftarrow 0]) < f_v(x[u \leftarrow 1])\), and \((uv, \ominus) \in E\) (i.e., \(u\) negatively affects the value of \(f_v\)) iff there is a state \(x\) such that \(f_v(x[u \leftarrow 0]) > f_v(x[u \leftarrow 1])\).


At each time step \(t\), variable \(v\) can update its state to \(s'(v) = f_v(s)\), where \(s\) (resp.\ \(s'\)) is the state of \(f\) at time \(t\) (resp.\ \(t + 1\)).
An \emph{update scheme} of a BN refers to how variables update their states over (discrete) time~\cite{schwab2020concepts}.
Various update schemes exist, but the primary types are \emph{synchronous}, where all variables update simultaneously, and \emph{fully asynchronous}, where a single variable is non-deterministically chosen for updating.
By adhering to the update scheme, the BN transitions from one state to another, which may or may not be the same.
This transition is referred to as the \emph{state transition}.
Then the dynamics of the BN is captured by a directed graph referred to as the State Transition Graph (STG).
We use \(\stg{f}\) (resp.\ \(\atg{f}\)) to denote the synchronous (resp.\ asynchronous) STG of \(f\).

A non-empty set of states is a \emph{trap set} if it has no out-going arcs on the STG of \(f\).
An \emph{attractor} is a subset-minimal trap set.
An attractor of size 1 (resp.\ at least 2) is called a fixed point (resp.\ cyclic attractor).
A \emph{sub-space} \(m\) of a BN is a mapping \(m \colon \var{f} \mapsto \mathbb{B} \cup \{\star\}\).
A sub-space \(m\) is equivalent to the set of all states \(s\) such that \(s(v) = m(v), \forall v \in \var{f}, m(v) \neq \star\).
With abuse of notation, we use \(m\) and its equivalent set of states interchangeably.
For example, \(m = \{v_1 = \star, v_2 = 1, v_3 = 1\} = \{011, 111\}\) (for simplicity, we write states as a sequence of values).
If a sub-space is also a trap set, it is a \emph{trap space}.
Unlike trap sets and attractors, trap spaces of a BN are independent of the update scheme~\cite{DBLP:journals/nc/KlarnerBS15}.
Then a trap space \(m\) is minimal iff there is no other trap space \(m'\) such that \(m' \subset m\).
It is easy to derive that a minimal trap space contains at least one attractor of the BN regardless of the update scheme.

\begin{example}\label{example:bn}
	Consider BN \(f_1\) with \(f_p = \neg q, f_q = \neg p, f_r = q\).
	Figures~\ref{figure:influence-sync-async-transition-graphs}~(a), (b), and (c) show the influence graph, the synchronous STG, and the asynchronous STG of \(f_1\).
	Attractor states are highlighted with boxes.
	\(\stg{f_1}\) has two fixed points and one cyclic attractor, whereas \(\atg{f_1}\) has only two fixed points.
	\(f_1\) has five trap spaces: \(m_1 = 10\star\), \(m_2 = 01\star\), \(m_3 = \star\star\star\), \(m_4 = 100\),  and \(m_5 = 011\).
	Among them, \(m_4\) and \(m_5\) are minimal.
\end{example}

\begin{figure}[!ht]
	\centering
	\begin{subfigure}{0.2\textwidth}
		\centering
		\begin{tikzpicture}[node distance=1cm and 1cm, every node/.style={scale=0.8}, line width = 0.5mm]
		\node[circle, draw] (p) [] {$p$};
		\node[circle, draw] (q) [below=of p] {$q$};
		\node[circle, draw] (r) [below=of q] {$r$};
		
		\draw[->] (q) edge [bend right=30] node [midway, above, fill=white] {$\ominus$} (p);
		\draw[->] (p) edge [bend right=30] node [midway, above, fill=white] {$\ominus$} (q);
		\draw[->] (q) edge [] node [midway, above, fill=white] {$\oplus$} (r);
		\end{tikzpicture}
		\caption{}
	\end{subfigure}
	\begin{subfigure}{0.35\textwidth}  
		\centering
		\begin{tikzpicture}[node distance=1cm and 1cm, every node/.style={scale=0.8}, line width = 0.5mm]
		\node[] (pr) [] {101};
		\node[draw] (p) [below=of pr] {100};
		\node[] (q) [right=of pr] {010};
		\node[draw] (qr) [below=of q] {011};
		\node[] (e) [below=of p] {000};
		\node[] (pqr) [below=of qr] {111};
		\node[draw] (r) [right=of pqr] {001};
		\node[draw] (pq) [right=of qr] {110};
		
		\draw[->] (pr) edge [] (p);
		\draw[->] (q) edge [] (qr);
		\draw[->] (e) edge [] (pq);
		\draw[->] (pqr) edge [] (r);
		\draw[->] (pq) edge [bend left=15] (r);
		\draw[->] (r) edge [bend left=15] (pq);
		
		\path [] (p) edge [out=0, in=45, loop] (p);
		\path [] (qr) edge [out=0, in=45, loop] (qr);
		\end{tikzpicture}
		\caption{}
	\end{subfigure}
	\begin{subfigure}{0.4\textwidth}  
		\centering
		\begin{tikzpicture}[node distance=1cm and 1cm, every node/.style={scale=0.8}, line width = 0.5mm]
		\node[] (pr) [] {101};
		\node[draw] (p) [below=of pr] {100};
		\node[] (q) [right=of pr] {010};
		\node[draw] (qr) [below=of q] {011};
		\node[] (e) [below=of p] {000};
		\node[] (pqr) [below=of qr] {111};
		\node[] (r) [right=of pqr] {001};
		\node[] (pq) [right=of qr] {110};
		
		\draw[->] (pr) edge [] (p);
		\draw[->] (q) edge [] (qr);
		
		\draw[->] (e) edge [] (p);
		\draw[->] (e) edge [] (q);
		
		\draw[->] (pqr) edge [] (qr);
		\draw[->] (pqr) edge [bend right=15] (pr);
		
		\draw[->] (pq) edge [bend left=15] (q);
		\draw[->] (pq) edge [bend left=25] (p);
		\draw[->] (pq) edge [bend left=15] (pqr);
		
		\draw[->] (r) edge [bend left=25] (pr);
		\draw[->] (r) edge [bend left=15] (qr);
		\draw[->] (r) edge [bend left=15] (e);
		
		\path [] (p) edge [out=180, in=135, loop] (p);
		\path [] (qr) edge [out=0, in=45, loop] (qr);
		
		\path [] (pr) edge [out=225, in=180, loop] (pr);
		\path [] (q) edge [out=-45, in=0, loop] (q);
		\end{tikzpicture}
		\caption{}
	\end{subfigure}
	\caption{(a) \(\ig{f_1}\), (b) \(\stg{f_1}\), and (c) \(\atg{f_1}\).}\label{figure:influence-sync-async-transition-graphs}
\end{figure}

\section{Finite ground normal logic programs and Boolean networks}\label{sec:connection}

We define a BN encoding for finite ground programs in Definition~\ref{definition:BN-encoding}.
Then, we show two relationships between a finite ground program and its encoded BN (see Theorems~\ref{theorem:ig-dg} and~\ref{theorem:lp-bn-supported-ts}).

\begin{definition}\label{definition:BN-encoding}
	Let \(P\) be a finite ground program. We define a BN \(f\) encoding \(P\) as follows: \(\var{f} = \atom{P}\), \(f_v = \bigvee_{r {\in} P, v = \head{r}}\bodyf{r}, \forall v \in \text{var}_f\).
	Conventionally, if there is no rule \(r \in P\) such that \(\head{r} = v\), then \(f_v = 0\).
	By considering 1 (resp.\ 0) as \textbf{t} (resp.\ \textbf{f}),  and \(\star\) as \(\textbf{u}\), sub-spaces (resp.\ states) of \(f\) are identical to 3-valued (resp.\ 2-valued) interpretations of \(P\).
\end{definition}

\begin{theorem}\label{theorem:ig-dg}
	Let \(P\) be a finite ground program and \(f\) be its encoded BN\@.
	Then \(\ig{f} \subseteq \dg{P}\).
\end{theorem}
\begin{proof}
	By construction, \(\ig{f}\) and \(\dg{P}\) have the same set of vertices.
	Let \(\text{in}_f^+(v)\) (resp.\ \(\text{in}_P^+(v)\)) denote the set of vertices \(u\) such that \((uv,\oplus)\) is an arc of \(\ig{f}\) (resp.\ \(\dg{P}\)).
	We define \(\text{in}_f^-(v)\) (resp.\ \(\text{in}_P^-(v)\)) similarly.
	We show that \(\text{in}_f^+(v) \subseteq \text{in}_P^+(v)\) and \(\text{in}_f^-(v) \subseteq \text{in}_P^-(v)\) for every \(v \in \atom{P}\) (*).
	Consider atom \(u\).
	The case that both \(u\) and \(\dng{u}\) appear in rules whose heads are \(v\) is trivial.
	For the case that only \(u\) appears in rules whose heads are \(v\), \(u\) is essential in \(f_v\) by construction, and it positively affects the value of \(f_v\), leading to \(u \in \text{in}_f^+(v)\) and \(u \not \in \text{in}_f^-(v)\).
	This implies that (*) still holds.
	The case that only \(\dng{u}\) appears in rules whose heads are \(v\) is similar.
	By (*), we can conclude that \(\ig{f} \subseteq \dg{P}\), i.e., \(\ig{f}\) is a sub-graph of \(\dg{P}\).
	In addition, if \(P\) is a quasi-interpretation finite ground program, then \(\ig{f} = \dg{P}\).
\end{proof}

\begin{lemma}[derived from Theorem 4.5 of~\cite{DBLP:conf/birthday/InoueS12}]\label{lemma:lp-bn-transition-graph}
	Let \(P\) be a finite ground program and \(f\) be its encoded BN\@.
	Then \(\tgsp{P} = \stg{f}\).
\end{lemma}

\begin{theorem}\label{theorem:lp-bn-supported-ts}
	Let \(P\) be a finite ground program and \(f\) be its encoded BN\@.
	Then supported trap spaces of \(P\) coincide with trap spaces of \(f\).
\end{theorem}
\begin{proof}
	By Lemma~\ref{lemma:lp-bn-transition-graph}, \(\tgsp{P} = \stg{f}\).
	Note that trap spaces of \(f\) are the same under both the synchronous and asynchronous update schemes~\cite{DBLP:journals/nc/KlarnerBS15}.
	Hence, trap spaces of \(f\) coincide with trap spaces of \(\stg{f}\).
	Since \(\tgsp{P} = \stg{f}\), supported trap spaces of \(P\) coincide with trap spaces of \(f\).
\end{proof}

For illustration, BN \(f_1\) of Example~\ref{example:bn} is the encoded BN of finite ground program \(P_1\) of Example~\ref{example:lp}.
\(\tgsp{P_1}\) is identical to \(\stg{f_1}\), and the five supported trap spaces of \(P_1\) are identical to the five trap spaces of \(f_1\).
In addition, \(P_1\) is tight and \(\ig{f_1} = \dg{P_1}\).

\section{Graphical analysis results}\label{sec:main-results}

In this section, we present our new results on graphical conditions for several properties of regular models in finite ground normal logic programs by exploiting the connection established in Section~\ref{sec:connection}.

\subsection{Preparations}
For convenience, we first recall several existing results in both logic programs and Boolean networks that shall be used later.

\begin{theorem}[\cite{DBLP:conf/birthday/InoueS12}]\label{theorem:quasi-transition-graph}
	Let \(P\) be a quasi-interpretation finite ground program.
	Then \(\tgst{P} = \tgsp{P}\), i.e., the stable and supported transition graphs of \(P\) are the same.
\end{theorem}

\begin{theorem}[\cite{DBLP:conf/birthday/InoueS12}]\label{theorem:lfp-transition-graph}
	Let \(P\) be a finite ground program and \(\lfp{P}\) denote its least fixpoint.
	Then \(P\) and \(\lfp{P}\) have the same stable transition graph.
\end{theorem}

\begin{theorem}[Theorem 6 of~\cite{TB24-static-analysis}]\label{theorem:lfp-no-cycle}
	Let \(P\) be a finite ground program and \(\lfp{P}\) denote its least fixpoint.
	If \(P\) is locally stratified, then \(\dg{\lfp{P}}\) has no cycle.
\end{theorem}

\begin{lemma}
	\label{lemma:no-neg-cycle}
	Let \(P\) be a finite ground program and \(\lfp{P}\) denote its least fixpoint.
	If \(\dg{P}\) is has no negative cycle, then \(\dg{\lfp{P}}\) has no negative cycle.
\end{lemma}
\begin{proof}
	It directly follows from Lemma 5.3 of~\cite{cois1994consistency}.
\end{proof}

\begin{proposition}[\cite{TBFS24lpnmr-semantics}]\label{prop:subset-supported-ts-clark}
	Let \(P\) be a finite ground program.
	Let \(T(P)\) denote the set of all supported trap spaces of \(P\).
	Let \(C(P)\) denote the set of all 3-valued models of \(\cf{P}\) (i.e., the Clark's completion of \(P\)).
	For every supported trap space \(I \in T(P)\), there is a model \(I' \in C(P)\) such that \(I' \leq_s I\).
\end{proposition}
\begin{proof}[Sketch of proof]
	Let \(I^j\) be an arbitrary supported trap space in \(T(P)\).
	We construct a 3-valued interpretation \(I^{j + 1}\) as follows: \(\forall a \in \atom{P}, I^{j + 1}(a) = I^j(\rhs{a})\).
	We prove that \(I^{j + 1}\) is also a supported trap space of \(P\).
	For every supported trap space \(I\) in \(T(P)\), we start with \(I^j = I\) and repeat the above process by increasing \(j\) by 1, and finally reach the case that \(I^{j + 1} = I^j\) because \(\gamma(I)\) is finite.
	By construction, \(I^j(a) = I^j(\rhs{a}), \forall a \in \atom{P}\), and \(I^j \leq_s I\).
	Hence, by setting \(I' = I^j\), there is a model \(I' \in C(P)\) such that \(I' \leq_s I\).
\end{proof}

\begin{theorem}[\cite{TBFS24lpnmr-semantics}]\label{theorem:regular-stable-ts}
	Let \(P\) be a finite ground program.
	Then a 3-valued interpretation \(I\) is a regular model of \(P\) iff \(I\) is a \(\leq_s\)-minimal stable trap space of \(P\).
\end{theorem}
\begin{proof}[Sketch of proof]
	Let \(\lfp{P}\) be the least fixpoint of \(P\).
	By Proposition~\ref{prop:subset-supported-ts-clark}, we can deduce that \(\leq_s\)-minimal supported trap spaces of \(\lfp{P}\) coincide with \(\leq_s\)-minimal supported (also stable) partial models spaces of \(\lfp{P}\).
	\(P\) and \(\lfp{P}\) have the same set of stable partial models~\cite{DBLP:journals/jlp/AravindanD95}.
	By Theorem~\ref{theorem:lfp-transition-graph}, \(P\) and \(\lfp{P}\) have the same stable transition graph, thus they have the same set of stable trap spaces.
	Since stable trap spaces of \(\lfp{P}\) coincide with its supported trap spaces, we can conclude the theorem.
\end{proof}

\begin{theorem}[Theorem 1 of~\cite{richard2019positive}]\label{theorem:bn-unique-att}
	Let \(f\) be a BN\@.
	If \(\ig{f}\) has no cycle, \(\atg{f}\) has a unique attractor that is also the unique fixed point of \(f\).
\end{theorem}

\begin{theorem}[Theorem 12 of~\cite{richard2019positive}]\label{theorem:no-neg-cycle-bn}
	Let \(f\) be a BN\@.
	If \(\ig{f}\) has no negative cycle, then \(\atg{f}\) has no cyclic attractor.
\end{theorem}

\subsection{Unicity of regular and stable models}

To illustrate better applications of the connection between finite ground programs and Boolean networks, we start with providing a probably simpler proof for the finite case of a well-known result on the unicity of regular and stable models in locally stratified programs~\cite{DBLP:journals/amai/EiterLS97}.

\begin{theorem}[\cite{DBLP:journals/amai/EiterLS97}]\label{theo:locally-stratified-program}
	If \(P\) is a locally stratified finite ground program, then \(P\) has a unique regular model that is also the unique stable model of \(P\).
\end{theorem}
\begin{proof}[New proof]
	Let \(\lfp{P}\) denote the least fixpoint of \(P\).
	Let \(f\) be the encoded BN of \(\lfp{P}\).
	By Theorem~\ref{theorem:lfp-no-cycle}, \(\dg{\lfp{P}}\) has no cycle.
	Since \(\ig{f}\) is a sub-graph of \(\dg{\lfp{P}}\) by Theorem~\ref{theorem:ig-dg}, it also has no cycle.
	By Theorem~\ref{theorem:bn-unique-att}, \(\atg{f}\) has a unique attractor that is also the unique fixed point of \(f\).
	\(P\) and \(\lfp{P}\) have the same set of regular (also stable) models~\cite{DBLP:journals/jlp/AravindanD95}.
	By Theorem~\ref{theorem:regular-stable-ts}, regular models of \(\lfp{P}\) are \(\leq_s\)-minimal stable trap spaces of \(\lfp{P}\).
	Since \(\lfp{P}\) is a quasi-interpretation finite ground program, its stable trap spaces coincide with its supported trap spaces.
	Supported trap spaces of \(\lfp{P}\) coincide with trap spaces of \(f\) by Theorem~\ref{theorem:lp-bn-supported-ts}.
	Hence, regular models of \(P\) coincide with \(\leq_s\)-minimal trap spaces of \(f\).
	Since the number of \(\leq_s\)-minimal trap spaces of \(f\) are a lower bound of the number of attractors of \(\atg{f}\) and \(f\) has at least one \(\leq_s\)-minimal trap space~\cite{DBLP:journals/nc/KlarnerBS15}, \(f\) has a unique \(\leq_s\)-minimal trap space that is also the unique fixed point of \(f\).
	Hence, \(P\) has a unique regular model that is also the unique stable model of \(P\).
\end{proof}

\subsection{Existence of non-trivial regular models}

\begin{theorem}[Theorem 5.3(i) of~\cite{DBLP:journals/jcss/YouY94}]\label{theorem:no-neg-cycle-all-fixed-well-stra}
	Let \(P\) be a well-founded stratification normal logic program.
	If \(\dg{P}\) has no negative cycle, then all the regular models of \(P\) are 2-valued.
\end{theorem}

Theorem~\ref{theorem:no-neg-cycle-all-fixed-well-stra} provides a sufficient (resp.\ necessary) condition on the dependency graph for the non-existence (resp.\ existence) of non-trivial regular models, but it is only limited to well-founded stratification normal logic programs.
Note that the well-founded stratification of a normal logic program is defined based on the ground instantiation of this program~\cite{DBLP:journals/jcss/YouY94}, and the set of all possible well-founded stratification programs in the finite case is only a small piece of the set of all possible finite ground programs~\cite{DBLP:journals/jcss/YouY94}.
To the best of our knowledge, the question if it is valid for any finite ground program is still open to date.
We answer this question in Theorem~\ref{theorem:no-neg-cycle-all-fixed-generic}.

\begin{theorem}\label{theorem:no-neg-cycle-all-fixed-generic}
	Let \(P\) be a finite ground program.
	If \(\dg{P}\) has no negative cycle, then all the regular models of \(P\) are 2-valued.
\end{theorem}
\begin{proof}
	Let \(\lfp{P}\) be the least fixpoint of \(P\).
	By Lemma~\ref{lemma:no-neg-cycle}, \(\dg{\lfp{P}}\) has no negative cycle.
	Let \(f\) be the encoded BN of \(\lfp{P}\).
	Since \(\ig{f}\) is a sub-graph of \(\dg{\lfp{P}}\) by Theorem~\ref{theorem:ig-dg}, \(\ig{f}\) also has no negative cycle.
	By Theorem~\ref{theorem:no-neg-cycle-bn}, \(\atg{f}\) (i.e., the asynchronous transition graph of \(f\)) has no cyclic attractor.
	This implies that all attractors of \(\atg{f}\) are fixed points (*).
	Assume that \(f\) has a \(\leq_s\)-minimal trap space (say \(m\)) that is not a fixed point.
	Since every \(\leq_s\)-minimal trap space of \(f\) contains at least one attractor of \(\atg{f}\)~\cite{DBLP:journals/nc/KlarnerBS15}, there is an attractor (say \(A\)) of \(\atg{f}\) such that \(A \subseteq \gamma(m)\).
	By (*), \(A\) is a fixed point, leading to \(A <_s m\).
	This is a contradiction because \(m\) is \(\leq_s\)-minimal.
	Hence, all \(\leq_s\)-minimal trap spaces of \(f\) are fixed points.
	
	By Theorem~\ref{theorem:lp-bn-supported-ts}, trap spaces of \(f\) coincide with supported trap spaces of \(\lfp{P}\).
	\(\lfp{P}\) is a quasi-interpretation finite ground program, thus \(\tgst{\lfp{P}} = \tgsp{\lfp{P}}\).
	It follows that its supported trap spaces are also its stable trap spaces.
	Hence, \(\leq_s\)-minimal trap spaces of \(f\) are \(\leq_s\)-minimal stable trap spaces of \(\lfp{P}\).
	This implies that all \(\leq_s\)-minimal stable trap spaces of \(\lfp{P}\) are 2-valued.
	By Theorem~\ref{theorem:regular-stable-ts}, all regular models of \(\lfp{P}\) are 2-valued.
	\(P\) and \(\lfp{P}\) have the same set of regular models~\cite{DBLP:journals/jlp/AravindanD95}.
	Hence, all regular models of \(P\) are 2-valued.
\end{proof}

Theorem~\ref{theorem:no-neg-cycle-all-fixed-generic} implies that the undefinedness is only needed if there is a negative cycle in the DG, i.e., the regular model and stable model semantics are the same under the absence of negative cycles.
In addition, we can get from Theorem~\ref{theorem:no-neg-cycle-all-fixed-generic} a straightforward corollary: if the DG of a finite ground program has no negative cycle, then it has at least one stable model.
The reason is because a finite ground program always has at least one regular model~\cite{DBLP:journals/jcss/YouY94}.
This corollary is exactly the generalization of the finite case of Theorem 5.7 of~\cite{DBLP:journals/jcss/YouY94} for well-founded stratification programs.

\subsection{Unicity of regular models}

The work of~\cite{DBLP:journals/jcss/YouY94} shows a sufficient condition for the unicity of regular models for well-founded stratification normal logic programs.

\begin{theorem}[Theorem 5.3(ii) of~\cite{DBLP:journals/jcss/YouY94}]\label{theorem:no-pos-cycle-unique-regular-well-stra}
	Let \(P\) be a well-founded stratification program.
	If \(\dg{P}\) has no positive cycle, \(P\) has a unique regular model.
\end{theorem}

Hereafter, we would like to show that the finite case of Theorem~\ref{theorem:no-pos-cycle-unique-regular-well-stra} is also true for any finite ground program.
Note however that the technique of using least fixpoint applied for negative cycles seems difficult to use for positive cycles because there is some finite ground program whose dependency graph has no positive cycle but the dependency graph of its least fixpoint can have positive cycle (e.g., \(P = \{a \leftarrow c; b \leftarrow c; c \leftarrow \dng{a}, \dng{b}\}\)).
We here use another approach.

\begin{theorem}[Theorem 3.4 of~\cite{DBLP:journals/entcs/PauleveR12}]\label{theorem:no-pos-cycle-unique-att}
	Let \(f\) be a BN\@.
	If \(\ig{f}\) has no positive cycle, then \(\atg{f}\) has a unique attractor.
\end{theorem}

\begin{theorem}[Lemma 16 of~\cite{DBLP:journals/jancl/DietzHW14}]\label{theorem:Clark-equal-stable-partial-model}
	Supported partial models of a tight finite ground program coincide with its stable partial models.
\end{theorem}

\begin{lemma}\label{lemma:tight-lp-regular-min-ts-bn}
	Let \(P\) be a finite ground program and \(f\) be its encoded BN\@.
	If \(P\) is tight, then regular models of \(P\) coincide with \(\leq_s\)-minimal trap spaces of \(f\).
\end{lemma}
\begin{proof}
	Since \(P\) is tight, stable partial models of \(P\) coincide with supported partial models of \(P\) (i.e., 3-valued models of \(\cf{P}\)) by Theorem~\ref{theorem:Clark-equal-stable-partial-model}.
	Then regular models of \(P\) coincide with \(\leq_s\)-minimal supported partial models of \(P\).
	We have that trap spaces of \(f\) coincide with supported trap spaces of \(P\) by Theorem~\ref{theorem:lp-bn-supported-ts}.
	By Proposition~\ref{prop:subset-supported-ts-clark}, \(\leq_s\)-minimal supported partial models of \(P\) coincide with \(\leq_s\)-minimal supported trap spaces of \(P\).
	Hence, regular models of \(P\) coincide with \(\leq_s\)-minimal trap spaces of \(f\).
\end{proof}

\begin{theorem}\label{theorem:no-pos-cycle-unique-regular-generic}
	Let \(P\) be a finite ground program.
	If \(\dg{P}\) has no positive cycle, then \(P\) has a unique regular model.
\end{theorem}
\begin{proof}
	Since \(\dg{P}\) has no positive cycle, \(\pdg{P}\) (i.e., the positive dependency graph of \(P\)) has no cycle, i.e., \(P\) is tight.
	Let \(f\) be the encoded BN of \(P\).
	By Lemma~\ref{lemma:tight-lp-regular-min-ts-bn}, regular models of \(P\) coincide with \(\leq_s\)-minimal trap spaces of \(f\).
	Since \(\ig{f}\) is a sub-graph of \(\dg{P}\), it also has no positive cycle.
	By Theorem~\ref{theorem:no-pos-cycle-unique-att}, \(\atg{f}\) has a unique attractor.
	Since every \(\leq_s\)-minimal trap space of \(f\) contains at least one attractor of \(\atg{f}\) and \(f\) has at least one \(\leq_s\)-minimal trap space~\cite{DBLP:journals/nc/KlarnerBS15}, \(f\) has a unique \(\leq_s\)-minimal trap space.
	Hence, we can conclude that \(P\) has a unique regular model.
\end{proof}

Since a stable model is also a regular model, Theorem~\ref{theorem:no-pos-cycle-unique-regular-generic} implies that if \(\dg{P}\) has no positive cycle, then \(P\) has at most one stable model.
In addition, \(P\) may have no stable model because the unique regular model may be not 2-valued.
This result seems to be already known in the folklore of the logic program theory, but to the best of our knowledge, there is no existing formal proof for it except the one that we have directly proved recently in~\cite{TB24-static-analysis}.

\subsection{Upper bound for number of regular models}

To the best of our knowledge, there is no existing work connecting between regular models of a finite ground program and (positive/negative) feedback vertex sets of its dependency graph.
In~\cite{TB24-static-analysis}, we have shown that \(2^{|U^{+}|}\) is an upper bound for the number of stable models where \(U^{+}\) is a positive feedback vertex set of the dependency graph.
Since stable models are 2-valued regular models, we can naturally generalize this result for the case of regular models, i.e., \(3^{|U^{+}|}\) is an upper bound for the number of regular models.
The underlying intuition for the base of three is that in a regular model, the value of an atom can be \textbf{t}, \textbf{f}, or \textbf{u}.

\begin{theorem}\label{theorem:regular-pfvs-three}
	Let \(P\) be a finite ground program.
	Let \(U^{+}\) be a positive feedback vertex set of \(\dg{P}\).
	Then the number of regular models of \(P\) is at most \(3^{|U^+|}\).
\end{theorem}
\begin{proof}
	By Theorem~\ref{theorem:regular-stable-ts}, regular models of \(P\) coincide with \(\leq_s\)-minimal stable trap spaces of \(P\).
	For any mapping \(\widehat{I} : U^{+} \to \{\textbf{t}, \textbf{f}, \textbf{u}\}\), we build a new finite ground program \(\widehat{P}\) from \(P\) as follows.
	First, remove from \(P\) all the rules whose heads belong to \(U^{+}\).
	Second, remove all the rules whose body formulas are false under the values of the atoms in \(U^{+}\) and otherwise remove all the appearances of the atoms that are in \(U^{+}\) and not assigned to \(\textbf{u}\) in \(\widehat{I}\).
	Third, for any atom \(a \in U^{+}\) such that \(\widehat{I}(a) = \textbf{u}\), add the rule \(a \leftarrow \dng{a}\).
	We can see that the part of \(\tgst{P}\) induced by \(\widehat{I}\) is isomorphic to \(\tgst{\widehat{P}}\).
	Hence, \(\leq_s\)-minimal stable trap spaces of \(P\) induced by \(\widehat{I}\) one-to-one correspond to those of \(\widehat{P}\).
	\(U^{+}\) intersects all positive cycles of \(\dg{P}\).
	Every atom \(a \in U^{+}\) such that \(\widehat{I}(a) \neq \textbf{u}\) is removed from \(\dg{P}\).
	In the case that \(a \in U^{+}\) and \(\widehat{I}(a) = \textbf{u}\), all the arcs ending at \(a\) are removed and an negative arc \((aa,\ominus)\) is added.
	It follows that \(\dg{\widehat{P}}\) has no positive cycle.
	By Theorem~\ref{theorem:no-pos-cycle-unique-regular-generic}, \(\widehat{P}\) has a unique \(\leq_s\)-minimal stable trap space.
	There are \(3^{|U^{+}|}\) possible mappings \(\widehat{I}\), thus we can conclude the theorem.
\end{proof}

\begin{theorem}[Theorem 3.5 of~\cite{DBLP:journals/entcs/PauleveR12}]\label{theorem:att-pfvs-two}
	Let \(f\) be a BN\@.
	Let \(U^{+}\) be a positive feedback vertex set of \(\ig{f}\).
	Then the number of attractors of \(\atg{f}\) is at most \(2^{|U^+|}\).
\end{theorem}

We observed that the bound of \(3^{|U^+|}\) is too rough for many example finite ground programs in the literature.
Then inspired by Theorem~\ref{theorem:att-pfvs-two} for an upper bound for the number of attractors of an asynchronous BN, we obtain an interesting result for tight finite ground programs.

\begin{theorem}\label{theorem:regular-pfvs-two-tight}
	Let \(P\) be a tight finite ground program.
	Let \(U^{+}\) be a positive feedback vertex set of \(\dg{P}\).
	Then the number of regular models of \(P\) is at most \(2^{|U^+|}\).
\end{theorem}
\begin{proof}
	Let \(f\) be the encoded BN of \(P\).
	By Lemma~\ref{lemma:tight-lp-regular-min-ts-bn}, regular models of \(P\) coincide with \(\leq_s\)-minimal trap spaces of \(f\).
	By definition, \(U^{+}\) intersects all positive cycles of \(\dg{P}\).
	Since \(\ig{f}\) is a sub-graph of \(\dg{P}\), every positive cycle of \(\ig{f}\) is also a positive cycle of \(\dg{P}\).
	Hence, \(U^{+}\) is also a positive feedback vertex set of \(\ig{f}\).
	By Theorem~\ref{theorem:att-pfvs-two}, the number of attractors of \(\atg{f}\) is at most \(2^{|U^+|}\).
	Since the number of \(\leq_s\)-minimal trap spaces of \(f\) is a lower bound of the number of attractors of \(\atg{f}\)~\cite{DBLP:journals/nc/KlarnerBS15}, the number of regular models of \(P\) is at most \(2^{|U^+|}\).
\end{proof}


\section{Conclusion and perspectives}\label{sec:conclusion}

In this paper, we have shown three main results relating some graphical properties of a finite ground normal logic program to the set of its regular models, namely 1) the presence of negative cycles as a necessary condition for the existence of non-trivial regular models, 2) the absence of positive cycles as a sufficient condition for the unicity of regular models, and 3) two upper bounds on the number of regular models for, respectively generic and tight, finite ground normal logic programs based on the size of positive feedback vertex sets in their dependency graph.
The first two conditions generalize the finite cases of the two existing results obtained by~\cite{DBLP:journals/jcss/YouY94} for well-founded stratification normal logic programs.
Our proofs use an encoding of finite ground normal logic programs by Boolean networks,
the equivalence established between regular models and minimal trap spaces, and some recent results obtained in Boolean network theory.

We believe that the established connection can provide more results for the study of Datalog programs and abstract argumentation, and might also be worth considering for normal logic programs without finiteness assumption on their ground intantiation.
The results presented in this paper use conditions on either positive cycles or negative cycles.
It is thus natural to think that by using both kinds of cycles simultaneously, improved results might be obtained.
Finally, we also conjecture that the upper bound for tight finite ground normal logic programs presented here, is in fact valid for generic ones.

\section*{Acknowledgments}

This work was supported by Institut Carnot STAR, Marseille, France.

\bibliographystyle{eptcs}
\bibliography{ref}
\end{document}